\documentclass[journal]{IEEEtran}

\ifCLASSINFOpdf
\else
   \usepackage[dvips]{graphicx}
\fi
\usepackage{url}

\usepackage{graphicx}
\usepackage{float}
\usepackage{algorithmic}
\usepackage{subfigure,multirow,bm,stfloats}
\usepackage{amsmath,amssymb,amsfonts}

\usepackage{algorithm}
\usepackage{xcolor}
\usepackage{soul}
\usepackage{amsthm}
\newtheorem{theorem}{Theorem}

\newtheorem{definition}{Definition}
\begin{document}

\title{Hypergraph Spectral Clustering for Point Cloud Segmentation}

\author{Songyang Zhang, Shuguang Cui, \IEEEmembership{Fellow, IEEE} and Zhi Ding, \IEEEmembership{Fellow, IEEE}
\thanks{S. Zhang and Z. Ding are with the Department of Electrical and Computer
	Engineering, University of California at Davis, Davis, CA 95616 USA (e-mail:
	sydzhang@ucdavis.edu; zding@ucdavis.edu).}
\thanks{S. Cui is currently with the Chinese University of Hong Kong, Shenzhen, China, 518172 (e-mail: shuguangcui@cuhk.edu.cn).}\vspace{-2em}
}

\maketitle

\begin{abstract}
Hypergraph spectral analysis has emerged as an
effective tool processing complex data structures in data
analysis. The surface of a {three-dimensional (3D)} point cloud and the multilateral 
relationship among their points can be naturally captured by the 
high-dimensional hyperedges. 
This work investigates the power of hypergraph spectral
analysis in unsupervised segmentation of 3D point clouds.
We estimate and order the hypergraph spectrum from observed 
point cloud coordinates.
By trimming the
redundancy from the estimated hypergraph spectral space 
based on spectral component strengths, we develop
a clustering-based segmentation method. 
We apply the proposed method to various point clouds, 
and analyze their respective spectral properties. Our experimental
results demonstrate the effectiveness 
and efficiency of the proposed segmentation method.
\end{abstract}


\IEEEpeerreviewmaketitle

\section{Introduction}

\IEEEPARstart{W}{ith} the proliferation of virtual reality (VR) and 
augmented reality (AR), 
three-dimensional (3D) point clouds have been widely
adopted as an efficient representation for 3D objects and their
surroundings in many applications \cite{c1}. 
One common processing task on 3D point clouds 
is unsupervised segmentation. The
goal of point cloud segmentation is to identify points in
a cloud with similar features for clustering 
into their respective regions \cite{c2}. 
These partitioned regions should be physically 
meaningful. Practical examples include the work of \cite{c3} 
which segments human posture point clouds for behavior analysis
by partitioning human bodies into different semantic body parts.
Segmentation facilitates point cloud analysis in various applications, 
such as object tracking, object
classification, feature extraction, and feature detection.

Among a myriad of methods proposed to segment point clouds, 
spectral clustering is commonly used in unsupervised scenarios \cite{c4,c5}.
By first modeling point clouds as graphs, one can derive
graph spectral space to apply spectral clustering 
in segmentation \cite{c6}. In \cite{c3}, the authors further applied
surface normals in spectral clustering to segment human bodies. 
Other graph-based spectral clustering methods include
graph partitioning \cite{c7} and spanning trees \cite{c8}.
Generally, these graph spectral clustering
methods can achieve good results and require less prior-konwledge 
on the datasets \cite{c2,c9}. 
Moreover, one can establish a natural connection between the point cloud features 
and the corresponding graph structure 
for analysis \cite{c10}. 
{However, despite their successes, graph-based methods 
still exhibit shortcomings, such as the low construction efficiency of graph models. 
Traditional graphs are constructed according to distance or similarity between data samples, 
e.g., using Guassian kernels \cite{gg1}. However, these constructions are purely based on observations without
prior knowledge of graph models, making them sensitive to noise with difficulty to tune the hyper-parameters \cite{gg}.
Moreover, even graph models from physical networks could be partially observed, the
construction of a suitable graph remains an open challenge.}
Furthermore, regular graph edges only connect two nodes and, thus, 
can only model pairwise relationships. 
Since surfaces in a point cloud usually contain more than two nodes,
point-to-point graph edges are ill-equipped to model such more complex
multilateral relationships. For these reasons, 
we are motivated to develop more general hypergraph models.

Hypergraphs are high-dimensional generalization of graphs.
Unlike edges in graphs that can only model the pairwise relationship of
two nodes, each hyperedge in a hypergraph connects more than two nodes.
The high-dimensionality of hyperedges can directly characterize 
the multilateral relationship of multiple points in a point cloud.
Thus, hypergraph can conveniently model point clouds. 
Moreover, motivated by graph signal processing (GSP) \cite{c11}, 
hypergraph-based signal processing tools can provide 
novel alternative definitions of hypergraph spectral space for 
spectral clustering \cite{c12,c13}. More specifically, in \cite{c12}, 
a hypergraph is constructed according to distances and the hypergraph 
spectrum is derived with tensor decomposition. However, such 
distance-based hypergraph construction is deficient
in measuring efficiency, and the tensor decomposition can also be time-consuming.

To overcome the aforementioned shortcomings, we 
propose a novel method of spectral clustering segmentation 
based on the hypergraph signal processing (HGSP) \cite{c12} 
for the gray-scale point clouds. 
We first 
propose to estimate the spectral components based on the hypergraph stationary 
process before ordering the components in accordance to their 
frequency coefficients. Removing information redundancy based on
spectrum order, a spectral clustering can be implemented on 
key spectrum components for point cloud segmentation. 
We test the proposed segmentation method on 
multiple gray-scale point cloud datasets to validate its effectiveness 
and efficiency. 
{Note that, since point cloud segmentation only requires a general order of the hypergraph spectrum to extract key information, it is insensitive to the noise. We approximate the frequency coefficients to reduce complexity while bypassing explicit denosing. Furthermore, experimental results validate the robustness of the proposed method in a noisy environment. However, many other applications may require explicit calculation of frequency coefficients and must overcome the noise effects. For these problems with different objectives, we present different methodologies in \cite{c18}}.

\section{Preliminary}
This section presents a brief overview on the fundamentals of HGSP \cite{c12}.
Within the HGSP framework, a hypergraph with $N$ nodes and the 
largest number of nodes $M$ connected by any hyperedge, is represented by an $M$th-order $N$-dimensional adjacency tensor $\mathbf{A}\in\mathbb{R}^{\underbrace{\scriptstyle{N\times N\times\cdots\times N}}_{\text{M times}}}$, which can be also decomposed via orthogonal {CANDECOMP/PARAFAC (CP)} decomposition \cite{c15}-\cite{c99}, i.e.,
$\mathbf{A}=(a_{i_1i_2\cdots i_M})\approx\sum_{r=1}^{N}\lambda_r\cdot\underbrace{\mathbf{f}_r\circ...\circ \mathbf{f}_r}_{\text{M times}}$,
where $\circ$ is the tensor outer product \cite{c12}, $\mathbf{f}_r$'s are orthonormal bases called spectral components, and $\lambda_r$'s are frequency coefficients related to hypergraph frequency. Note that the hyperedges with fewer nodes than $M$ are normalized with weights as 
described in \cite{c12,c14}.
With the definition of hypergraph spectrum, a supporting matrix 
$\mathbf{P_s}=\frac{1}{\lambda_{\max}}\mathbf{V}\bm{\Sigma}\mathbf{V}^T$ can be defined to capture the overall spectral information, where $\mathbf{V}=[\mathbf{f}_1,\cdots,\mathbf{f}_N]$ and $\mathbf{\Sigma}=diag(\lambda_r)\in\mathbb{R}^{N\times N}$. {We refer readers to \cite{c12} for additional discussions and properties of the supporting matrix.}

Because of page limitation, we shall refrain from elaborating in detail
many fundamental aspects of HGSP. Interested readers can find important concepts,
such as hypergraph Fourier transform, HGSP filter design, 
and sampling theory in \cite{c12}.

\section{Segmentation}
Our proposed segmentation method targets gray-scale point clouds
consisting of $N$ points. Such point clouds 
can be represented by $\mathbf{s}=[\mathbf{X_1\quad X_2\quad X_3}]\in\mathbb{R}^{N\times 3}$ 
where $\mathbf{X}_i$ captures the $N$-point positions
in the $i-$th coordinate. There are three stages in the 
proposed segmentation: 1) estimate the hypergraph spectral space, 2) order 
and select the principal hypergraph spectrum, 
and 3) segment via clustering in the reduced hypergraph spectral space. 
In the first stage,  instead of decomposing the constructed hypergraph, 
we estimate the hypergraph spectrum directly from observed point clouds based on the hypergraph stationary process. This approach by-passes explicit hypergraph construction since the representing tensor is memory-inefficient 
and its orthogonal-CP decomposition is time-consuming. 
We then estimate the distribution of hypergraph frequency coefficients 
according to a measure of smoothness, and order the spectrum based on the hypergraph frequency. 
Finally, we identify the low frequency spectral contents and cluster 
in the optimized spectral space.

\subsection{Estimation of Hypergraph Spectral Space}
We begin with the estimation of hypergraph spectrum.
In \cite{c16}, a graph stationary process is defined 
within the GSP framework to describe the stationary property of the graph shifting. 
Moreover, \cite{c17} proposes a method to estimate the graph spectral
components under the assumption of stationarity for the observed dataset. 
{For point cloud datasets, the three coordinates of each point can be interpreted as 
observations of a node from three different viewpoints, 
which reflects the structural information embedded in the spectrum.}
Thus, we can estimate the hypergraph spectrum components based on
hypergraph stationary processing.

{In \cite{c12}, a polynomial filter can be defined with a supporting matrix to capture the spectral information of adjacency tensor}. {Let $\mathbb{E}(\cdot)$ denote expectation and $(\cdot)^H$ denote conjugate transpose.
The hypergraph stationary process is defined as follows.}
\begin{definition}(Weak-Sense Stationary Process)
	A stochastic signal $\mathbf{x}\in\mathbb{R}^N$ is weak-sense stationary (WSS) over 
 hypergraph with supporting matrix 
	$\mathbf{P}_s$ {\textit{iff}} for all integers $\tau\ge 0$, 
	\begin{align}\label{mean}
	\mathbb{E}[\mathbf{x}]&=\mathbb{E}[\mathbf{P}_{\tau}\mathbf{x}]\\
	\label{time}
	\mathbb{E}[(\mathbf{P}_{\tau_1}\mathbf{x})((\mathbf{P}^H)_{\tau_2}\mathbf{x})^H] &=\mathbb{E}[(\mathbf{P}_{\tau_1+\tau}\mathbf{x})((\mathbf{P}^H)_{\tau_2-\tau}\mathbf{x})^H]
	\end{align}
	where $\mathbf{P}=\lambda_{max}\mathbf{P_s}$ and $\mathbf{P}_\tau=\mathbf{P}^{\tau}$.
\end{definition}

Condition (\ref{mean}) requires constant mean for stochastic signals over hypergraph, 
consistent with traditional definition of WSS stochastic processes.
Based on the transpose relationship, $\mathbf{P}^H$ can be interpreted as 
propagation in an opposite direction of $\mathbf{P}$. Hence, condition
(\ref{time}) implies that 
the covariance of stationary signals only depends on the difference 
between two steps, i.e., $\tau_1+\tau_2$.
With the definition of hypergraph stationary process, we have the following property.
\begin{theorem}
	A stochastic signal $\mathbf{x}$ is WSS if and only if it has zero-mean and its covariance 
	matrix has the same eigenvectors as the hypergraph spectrum basis, i.e., 
	$\mathbb{E}[\mathbf{x}] =\mathbf{0}$ and
	$\mathbb{E}[\mathbf{x}\mathbf{x}^H]=\mathbf{V}\Sigma_\mathbf{x}\mathbf{V}^{H}, \label{s2}$
	where $\mathbf{V}$ is the hypergraph spectrum.	
\end{theorem}

\begin{proof}
	Since the hypergraph spectrum basis are orthonormal, we have $\mathbf{VV}^T=\mathbf{I}$.
	Then, the $\tau$-step shifting based on supporting matrix can be calculated as
	\begin{align}
	\mathbf{P}_\tau&=\underbrace{\mathbf{V}\Lambda_\mathbf{P}\mathbf{V}^T\mathbf{V}\Lambda_\mathbf{P}\mathbf{V}^T\cdots\mathbf{V}\Lambda_\mathbf{P}\mathbf{V}^T}_{\tau\quad times}\\
	&=\mathbf{V}\Lambda_\mathbf{P}^\tau\mathbf{V}^T.\label{poly}
	\end{align}
	
	Now, the Eq. (\ref{mean}) can be written as 
	\begin{equation}
	\mathbb{E}[\mathbf{x}]=\mathbf{V}\Lambda_P^\tau\mathbf{V}^T\mathbb{E}[\mathbf{x}].
	\end{equation}
	Since $\mathbf{V}\Lambda_P^\tau\mathbf{V}^T$ does not always equal to $\mathbf{I}$, Eq. (\ref{mean}) holds for arbitrary supporting matrix and $\tau$ if and only if $\mathbb{E}[\mathbf{x}]=\mathbf{0}$. 
	
	Next we show the sufficiency and necessity of the condition in Eq. (\ref{s2}). 
	The condition in Eq. (\ref{time}) can be written as
	\begin{equation}
	\mathbf{P}_{\tau_1}\mathbb{E}[\mathbf{xx}^H]((\mathbf{P})^H_{\tau_2})^H=\mathbf{P}_{\tau_1+\tau}\mathbb{E}[\mathbf{xx}^H]((\mathbf{P})^H_{\tau_2-\tau})^H.
	\end{equation}
	Considering Eq. (\ref{poly}) and the fact that
	hypergraph spectrum is real, Eq. (\ref{time}) is equivalent to
	\begin{equation}
	\mathbf{V}\Lambda_\mathbf{P}^{\tau_1}\mathbf{V}^H\mathbb{E}[\mathbf{xx}^H]\mathbf{V}\Lambda_\mathbf{P}^{\tau_2}\mathbf{V}^H=\mathbf{V}\Lambda_\mathbf{P}^{\tau_1+\tau}\mathbf{V}^H\mathbb{E}[\mathbf{xx}^H]\mathbf{V}\Lambda_\mathbf{P}^{\tau_2-\tau} \mathbf{V}^H,
	\end{equation}
	which can be written as
	\begin{equation} \label{eig}
	(\mathbf{V}^H\mathbb{E}[\mathbf{xx}^H]\mathbf{V})\Lambda_\mathbf{P}^{\tau}=\Lambda_\mathbf{P}^{\tau}(\mathbf{V}^H\mathbb{E}[\mathbf{xx}^H]\mathbf{V}).
	\end{equation}
	If Eq. (\ref{eig}) holds for arbitrary $\mathbf{P}$, 	$(\mathbf{V}^H\mathbb{E}[\mathbf{xx}^H]\mathbf{V})$ should be diagonal, 
	which indicates $\mathbb{E}[\mathbf{x}\mathbf{x}^H]=\mathbf{V}\Sigma_\mathbf{x}\mathbf{V}^{H}$. 
	Thus, the sufficiency of the condition is proved. 
	
	Similarly, we can apply Eq. (\ref{s2}) on both sides of Eq. ($\ref{time}$), 
	we can establish the necessity of the condition in Eq. (\ref{s2}). 
\end{proof}
This property indicates that, we can estimate the hypergraph 
spectral components from the eigenspace of the covariance matrix. 
Accordingly, a hypergraph-based spectrum estimation can be developed. 
Given a gray-scale point cloud 
$\mathbf{s}=[\mathbf{X_1\quad X_2\quad X_3}]\in\mathbb{R}^{N\times 3}$, 
we can treat each $\mathbf{X}_i$ as an 
observation of the point data and normalize them to zero-mean. 
Assuming signals to be hypergraph WSS, through normalization of observations we can directly obtain the hypergraph 
spectrum from their covariance matrix.

\subsection{Estimation of the Spectrum Distribution}
One important issue in spectral clustering is the ranking 
of the spectral components in order to identify and remove 
some less critical and redundant information. 
Within the framework of HGSP, we rank the spectral components according 
to their
(nonnegative) frequency coefficients in descending order 
to relatively order spectral components from  low frequency to high 
frequency \cite{c12}. 
Clearly, the problem lies in estimating the spectrum distribution. 
In practical applications, large-scale networks are often sparse, 
thereby making it meaningful to infer that most entries of the 
hypergraph representing tensor 
in typical datasets are zero \cite{c19}. 
In addition, signal smoothness is a widely-used assumption 
when estimating the underlying structure of graphs and hypergraphs 
\cite{c20}. We formulate a general estimation of hypergraph coefficients 
as 
\begin{align}
&\min_{\mathbf{\boldsymbol{\lambda}}}\quad  \mbox{Smooth}
(\mathbf{s,\boldsymbol{\lambda}},\mathbf{f}_r)+\beta||\mathbf{A}||_T^2\label{e1}\\
s.t.\quad &\mathbf{A}=\sum_{r=1}^{N}\lambda_r\cdot\underbrace{\mathbf{f}_r\circ...\circ \mathbf{f}_r}_{\text{M times}}, \quad\mathbf{A}\in \mathcal{A}. \label{dec}
\end{align}
Here, $||\mathbf{A}||^2_T=\boldsymbol{\lambda}^T\boldsymbol{\lambda}$ 
is the tensor norm \cite{c18}, and the set $\mathcal{A}$ includes the
prior information on the tensor types, e.g., adjacency or Laplacian. 
{The optimization problem in Eq. (3) is similar to the traditional hypergraph learning framework \cite{c31,c32}, i.e., $\min_f R_{emp}(f)+\mu \Omega(f)$, where $R_{emp}$ is the empirical loss and $\Omega$ is a regularizer on the hypergraph. The major differences lie in the choice of loss functions.} {Specifically, we use a tensor-based
	smoothness function and energy regularizer, whereas many traditional hypergraph learning works typically focus on the matrix-based Laplacian \cite{c33,c34}}.

Instead of the exact calculation of frequency coefficients, 
our 3D point cloud segmentation only requires a 
general idea on the distribution of frequency coefficients. Thus, 
we can simplify the problem as follows. 
First, we limit our tensor order to $M=3$, i.e.,
each hyperedge has girth 3, since 
3 nodes are the required minimum to construct a surface. 
We then use the total variation based on supporting matrix, 
denoted by $\mathbf{TV}(\mathbf{s})=||\mathbf{s}-\mathbf{P_s}\mathbf{s}||^2_2$, 
to describe the smoothness over the estimated hypergraph. 
In addition, we set the first eigenvector in the covariance matrix of observations as the spectrum component corresponding to $\lambda_{\max}$ to maintain the information of the observed signals.

 Let $\boldsymbol{\sigma}=\boldsymbol{\lambda}/{\lambda_{max}}=[\sigma_1\quad \sigma_2 \quad \cdots\quad \sigma_N]^T$. 
 The formulation to estimate the spectrum distribution can be rewritten as 
\begin{equation}\label{target}{
\min_{\boldsymbol{\sigma}}\quad \sum_{i=1}^3||\mathbf{X_i}-\mathbf{P_s}\mathbf{X_i}||^2_2+\beta{ \boldsymbol\sigma^T\boldsymbol\sigma}}
\end{equation}
\begin{align}
s.t. \quad 
&0\leq \sigma_r\leq \sigma_{1}=1;\label{non}\\
&\sum_{r=1}^N \sigma_r f_{r,i_1}f_{r,i_2}f_{r,i_3}\geq 0, \quad  i_1,i_2,i_3=1,\cdots,N. \label{adj}
\end{align}
Note that the constraint (\ref{adj}) indicates that
the tensor $\mathbf{A}$ is an adjacency tensor here, and
can be modified or relaxed for specific applications. The constraint (\ref{non}) is the 
the nonnegative constraints on the factor matrices \cite{c15}.
Thus, the formulation is convex and can be readily solved by using numerical recipes.

\subsection{Segmentation based on Hypergraph Spectral Clustering}\label{f11}
With the estimated spectrum components and frequency coefficients, we can 
directly propose a segmentation method based on spectral clustering. 
The detailed steps are summarized as Alogrithm 1. Usually, 
we can define a threshold in Step 7. We will provide more information
on selecting the leading components in Section \ref{graphm}.

\begin{algorithm}[b]
	\begin{algorithmic}[1] 
		\caption{Hypergraph Spectral Clustering}\label{freestimation}
		\STATE {\bf{Input}}: Point cloud dataset $\mathbf{s}=[\mathbf{X_1,X_2,X_3}]\in\mathbb{R}^{N\times 3}$ and the number of clusters $k$.
		\STATE Calculate the mean of each row in $\mathbf{s}$, i.e., $\mathbf{\overline s}=(\mathbf{X_1+X_2+X_3})/3$;
		\STATE Normalize the original point cloud data to zero-mean in each row, i.e.,
		$\mathbf{s}'=[\mathbf{X_1-{\overline s},X_2-{\overline s},X_3-{\overline s}}]$;
		\STATE Calculate eigenvectors $\{\mathbf{f}_1,\cdots,\mathbf{f}_N\}$ for $R(\mathbf{s}')=\mathbf{s'}(\mathbf{s'}^H)$;
		\STATE Estimate frequency coefficients $\sigma_r$'s by solving Eq. (\ref{target});
		\STATE Rank frequency components $\mathbf{f}_r$'s based on their corresponding frequency coefficients $\sigma_r$ in the decreasing order.
		\STATE Find the first $E$ leading spectral components $\mathbf{f}_r$ with larger $\sigma_r$ and construct a spectrum matrix $\mathbf{M}\in\mathbb{R}^{N\times E}$ with columns as the leading spectrum components.
		\STATE Cluster the rows of $\mathbf{M}$ using $k$-means clustering.
		\STATE Cluster node $i$ into partition $j$ if the $i$th row of $\mathbf{M}$ is assigned to $j$th cluster.
		\STATE {\bf{Output}}: $k$ partitions of the point clouds.
	\end{algorithmic}
\end{algorithm}

\section{Experiments} \label{graphm}
We test the performance of the proposed method 
along with traditional graph-based methods and k-means clustering.

\begin{figure*}[t]
	\centering
	\subfigure[HGSP]{
		\label{mw1}
		\includegraphics[height=7cm]{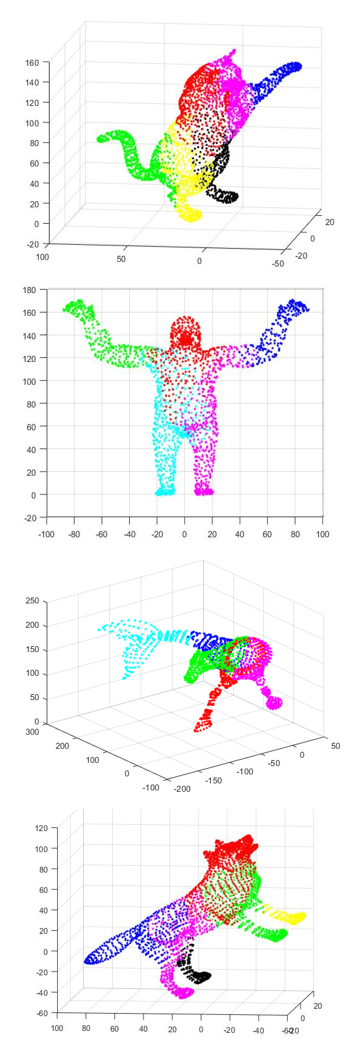}}
	\subfigure[GSP]{
		\label{mw2}
		\includegraphics[height=7cm]{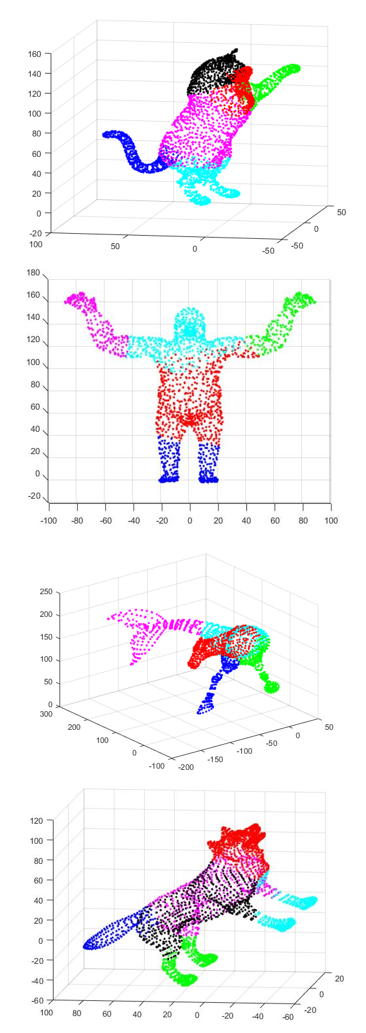}}
	\subfigure[Laplacian]{
		\label{mw3}
		\includegraphics[height=7cm]{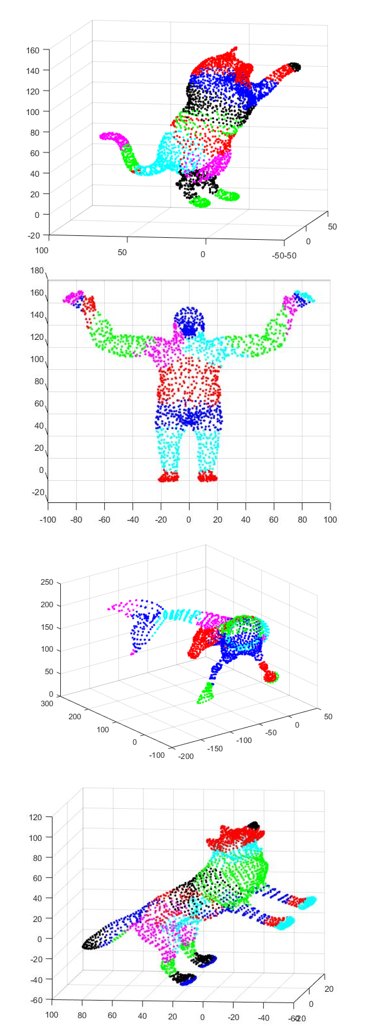}}
	\subfigure[Kmeans]{
		\label{mw4}
		\includegraphics[height=7cm]{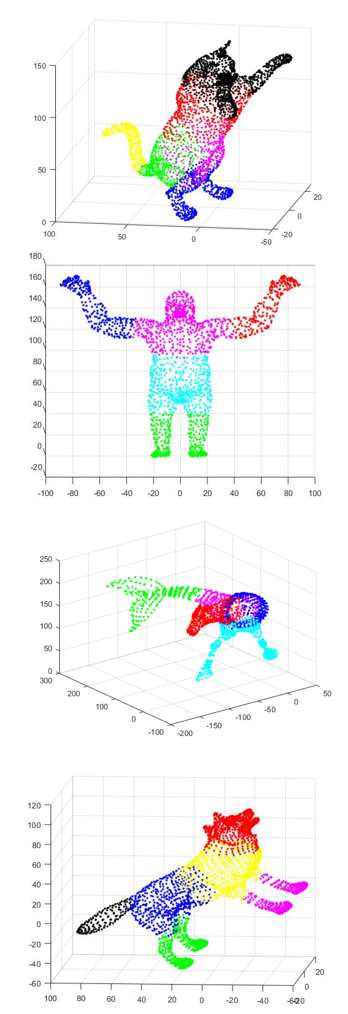}}
	\caption{Results of Segmentation.}
	\label{mw}
\end{figure*}

\textbf{Experiment Setup:} To implement $k$-means, we cluster over each row of the point 
cloud coordinates. For the hypergraph spectral clustering, {we select the first 
$E$ key spectrum components according to frequency coefficients
until a steep drop to the next (i.e., the $(E+1)$-th) 
coefficient.
Typically, the first two or three elements sufficiently satisfy this 
criterion. When optimizing the frequency coefficients, an efficient $\beta$ lies in $0.1$-$10$ depending on the specific datasets.} For the graph-based clustering,
a Gaussian-graph model \cite{c10} is applied
to encode the local geometry information through an adjacency matrix 
$\mathbf{W}\in\mathbb{R}^{N\times N}$. Let $\mathbf{s}_i\in\mathbb{R}^{1\times3}$ 
be the $i$th point coordinate. The edge weight between points $i$ and $j$ 
is calculated as $W_{ij}=\exp\left(-\frac{||\mathbf{s}_i-\mathbf{s}_j||^2_2}{\delta^2}
\right)$ if $\quad||\mathbf{s}_i-\mathbf{s}_j||^2_2\leq t$; otherwise, $W_{ij}=0$.
Here, the variance $\delta$ and the threshold $t$ are parameters to control 
the edge weights.  GSP spectrum is derived from the matrix $\mathbf{W}$. 
We also test the Laplacian matrix $\mathbf{L=D-S}$, where $\mathbf{S}$ is the unweighted adjacency matrix and $\mathbf{D}$ is the diagonal matrix of node degree,

\textbf{Overall Performance}: We first compare different methods in the animal datasets in \cite{c21}-\cite{c24}. The overall results are shown in Fig. \ref{mw}.
The test results show that HGSP-based method, GSP-based method, and k-means clustering 
exhibit similar performance by clustering limbs and torsos. 
Interestingly, our HGSP method can further distinguish tails and different 
legs.
Especially for the gorilla dataset in the second row of Fig. \ref{mw}, 
HGSP spectral clustering segments
different limbs with four different colors, whereas other methods fail to do so.
We can see that the hypergraph model captures the overall structural information of
3D point cloud better than traditional graphs. 
The Laplacian-based method accentuates the details of some complex structures. 
For example, in the gorilla dataset, 
Laplacian-based method further distinguishes feet from legs and hands from 
arms, respectively. 
Generally, HGSP-based spectral clustering presents clearer segmentation of the main features
for the point cloud datasets.
\begin{figure}[t]
	\centering
	\subfigure[400 Samples.]{
		\label{r1}
		\includegraphics[height=2cm]{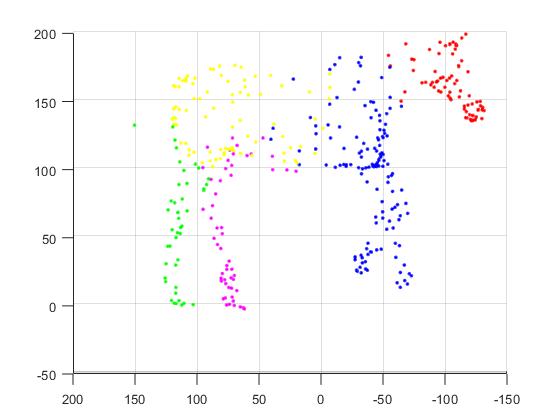}}
	\hfill
	\subfigure[1400 Samples.]{
		\label{r2}
		\includegraphics[height=2cm]{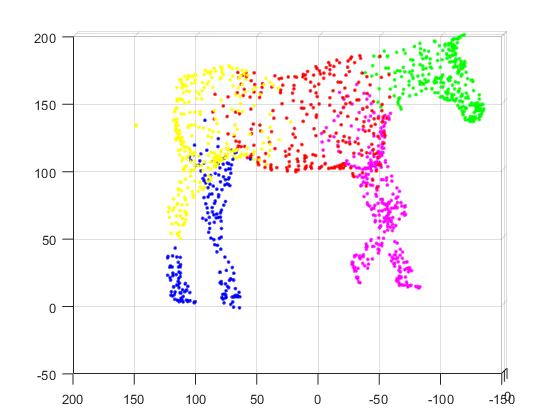}}
	\hfill
	\subfigure[3400 Samples.]{
		\label{r3}
		\includegraphics[height=2cm]{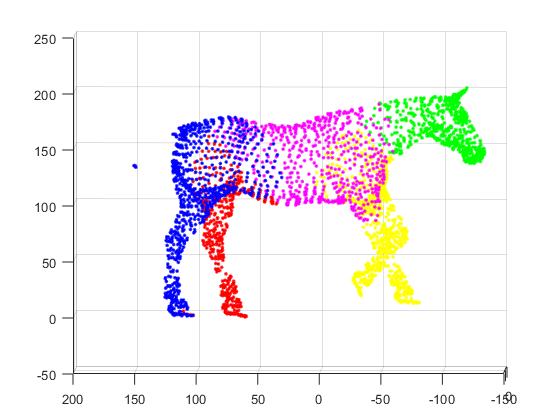}}
	\caption{HGSP Segmentation with Different Samples.}
	\label{rr1}
\end{figure}

\begin{figure}[t]
	\vspace{-0.0cm}
	\centering
	\subfigure[Clean Point Cloud.]{
		\label{s1}
		\includegraphics[height=2cm]{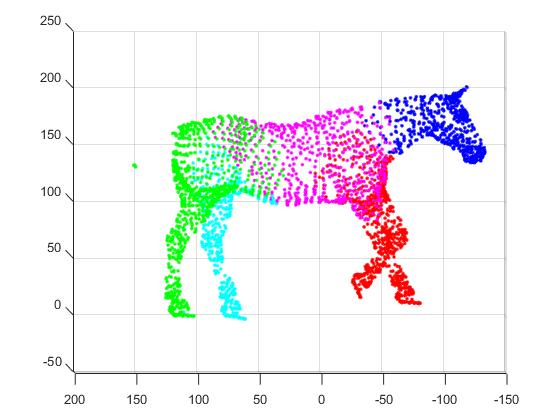}}
	\hfill
	\subfigure[With SNR=32 dB]{
		\label{ss2}
		\includegraphics[height=2cm]{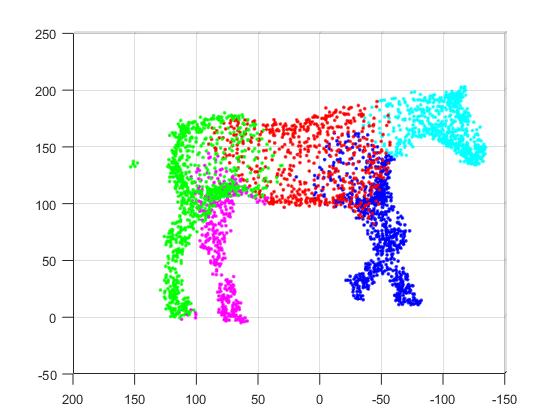}}
	\hfill
	\subfigure[With SNR=25 dB.]{
		\label{s3}
		\includegraphics[height=2cm]{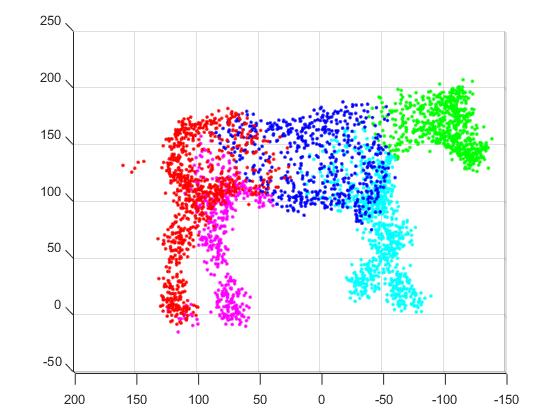}}
	\caption{HGSP Segmentation under Different Noise.}
	\label{sr1}
\end{figure}
\begin{table}[b]
	\scriptsize
	\caption{{Running Time of Different Methods (in Seconds)}}
	\begin{tabular}{|c|c|c|c|}
		\hline
		& Gorilla(2048 nodes)& Wolf(3400 nodes) & Cat(3400 nodes)  \\ \hline
		GSP       & 7.05                           & 24.982           & 24.812                      \\ \hline
		HGSP      & 2.771                          & 11.451           & 11.335                      \\ \hline
		Laplacian & 4.662                       & 15.579           & 15.773                     \\ \hline
		$k$-Means   & 0.016                            & 0.014            & 0.013                       \\ \hline
	\end{tabular}
	\\
	\centering
	\label{T_run}
\end{table}
\textbf{Numerical Comparison}: To provide comprehensive numerical comparison 
between different methods, we also compare the Silhouette index and mean 
accuracy of different methods in the ShapeNet Datasets \cite{c25,c26}. 
In the ShapeNet datasets, there are 16 categories of objects with labels 
in 2-6 classes. We test the average Silhouette and mean accuracy by 
randomly 
picking 50 point clouds from each category. 
The result is shown in Table \ref{t1}. From the result, 
we can see that the HGSP-based method provides the largest Silhouette indices
(indicating the best inner-cluster fitting) and the highest 
mean accuracy. Although these numerical results are valuable, 
larger mean accuracy does not necessarily imply better performance
in unsupervised clustering. 
For example, in Fig. \ref{plane1}, although the segmentation 
results differ from ground truth, these results still 
make sense by grouping two wings to different classes. 
Often, visualization can be a more suitable performance assessment.
Additional visualized results could be found in Fig. \ref{mw111}.

\begin{table}[t]
	\vspace{-4mm}
	\centering
	\caption{{Comparison in ShapeNet Datasets}}
	\vspace{-2mm}
	\begin{tabular}{|l|l|l|l|l|}
		\hline
		& HGSP                                      & GSP       & Laplacian & K-means    \\ \hline
		Silhouette & \textbf{0.56748}  & 0.25756 & 0.137381  & 0.55894    \\ \hline
		Accuracy   & \textbf{0.58928} & 0.55321 & 0.502275  & 0.57699 \\ \hline
	\end{tabular}
\label{t1}
\end{table}

\begin{figure}[t]
	\centering
	\subfigure[Ground Truth.]{
		\label{p1}
		\includegraphics[height=2cm]{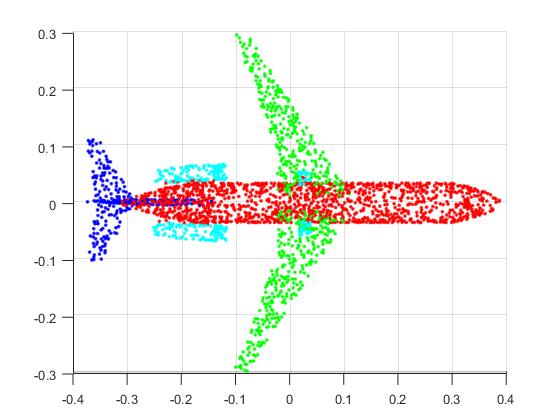}}
	\hspace{0.6in}
	\subfigure[HGSP Results.]{
		\label{p2}
		\includegraphics[height=2cm]{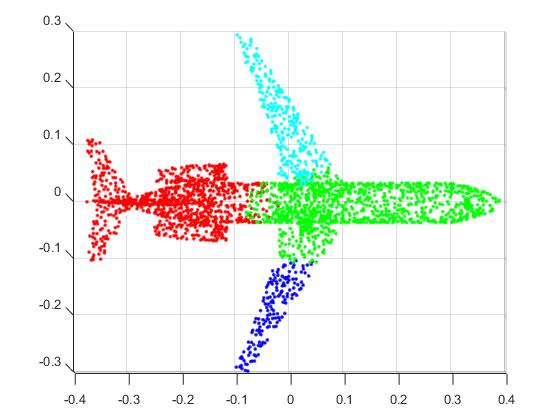}}
	\caption{Segmentation and Ground Truth.}
	\label{plane1}\vspace{-6mm}
\end{figure}


\begin{figure}[t]
	\centering
	\begin{minipage}[t]{4cm}
		\centering
		\includegraphics[width=1.3in]{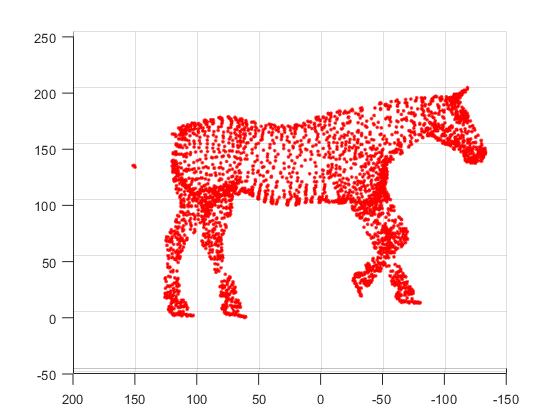}
		\caption{Horse Point Cloud.}
		\label{hor}
	\end{minipage}
	\hspace{0.1cm}
	\begin{minipage}[t]{4cm}
		\centering
		\includegraphics[width=1.3in]{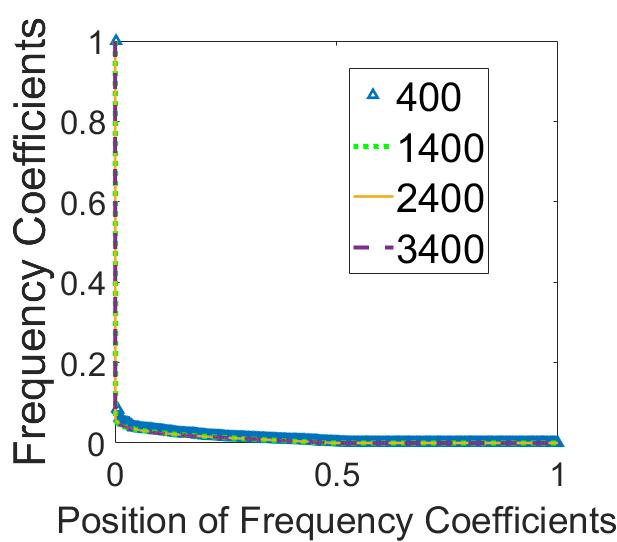}
		\caption{HGSP Coefficients.}
		\label{fhgsp}
	\end{minipage}
	
	\begin{minipage}[t]{4cm}
		\centering
		\includegraphics[width=1.3in]{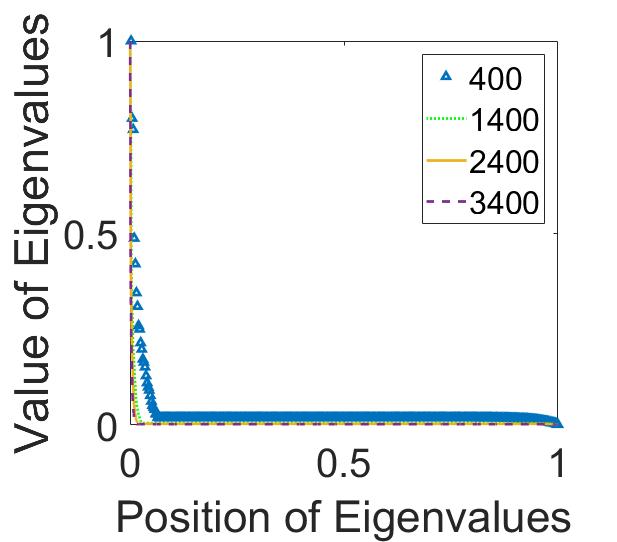}
		\caption{GSP Eigenvalues.}
		\label{fgsp}
	\end{minipage}
	\hspace{0.1cm}
	\begin{minipage}[t]{4cm}
		\centering
		\includegraphics[width=1.3in]{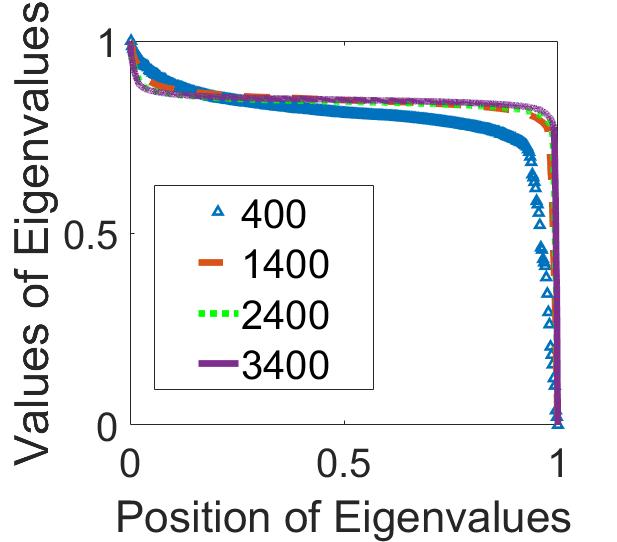}
		\caption{Laplacian Eigenvalues.}
		\label{flapla}
	\end{minipage}
\end{figure}
\textbf{Distribution of Eigenvalues}: We are interested in the reasons behind the performance differences of different graphical methods. 
To explore the reasons behind such differences, 
we examine the 
distributions of eigenvalues or the frequency coefficients of different methods 
in the specific horse point cloud shown in Fig. \ref{hor}. 
In different rounds of the experiment, we randomly sample 400, 1400, 2400 and 3400 points 
from the original horse point cloud and calculate the eigenvalues from different methods. 
The results are shown in Fig. \ref{fhgsp}, Fig. \ref{fgsp} and Fig. \ref{flapla}. 
The Y-axis is the normalized eigenvalues or frequency coefficients.
The X-axis is the eigenvalue order, i.e., $Pos_i=i/N$ for the $i$th 
eigenvalue of $N$ nodes. From the results, we can see that the 
HGSP-based method and GSP-based 
method have quite similar distributions, which indicate that their 
feature information is more 
concentrated in the first few key spectral components.{Moreover, the HGSP-based method delivers a sharper curve than the GSP-based method. As mentioned in \cite{c27}, a larger eigengap would lead to better clustering results, which should be responsible for the performance
difference between the HGSP-based and GSP-based methods.}
{Unlike adjacency-based methods, the distribution of eigenvalues of the 
Laplacian is rather different as shown in Fig. \ref{flapla}.} {In contrast to the Laplacian-based method, the HGSP-based method makes it easier to identify the vital information.
This difference in eigenvalue distribution can account for the performance
difference between the Laplacian-based segmentation and those based on adjacency.}

\textbf{Complexity and Robustness}:
We also test on datasets for different numbers of samples and noise effect. 
The results are shown in Fig. \ref{rr1} and Fig. \ref{sr1}. The 
HGSP spectral clustering remains robust for either noisy data or down-sampled data.
We compare the computation runtime of different methods 
over the animal datasets.  From results summarized in Table \ref{T_run}, 
it is not surprising
that the $k$-means method is the fastest, since graph-based methods 
require the additional step of spectrum estimation before clustering. 
The GSP-based and Laplacian-based methods require more computation, primarily because the computations needed to form the graph structure, whereas our proposed method 
directly estimates the HGSP spectral components.  In particular, we only require an approximate
distribution of the frequency coefficients to complete the segmentation task. Since the power of estimated coefficients is mainly concentrated in the first few spectral components shown as the optimized distribution in Fig. \ref{fhgsp}, a faster implementation can be done with the knowledge of the key estimated hypergraph spectra.

\section{Conclusion}
This work proposes a novel segmentation method for 3D point clouds 
based on hypergraph spectral clustering. We first estimate the hypergraph spectral space 
via hypergraph stationary processing
before ranking the spectral components according to their frequency coefficients. 
We further introduce a robust segmentation algorithm that utilizes the estimated 
hypergraph spectrum pairs. The test results over multiple point cloud datasets 
clearly demonstrate the advantages of the proposed method and the power of 
HGSP in 3D point clouds.

\begin{figure}[t]
	\centering
	\subfigure[HGSP]{
		\label{mw11}
		\includegraphics[height=13cm]{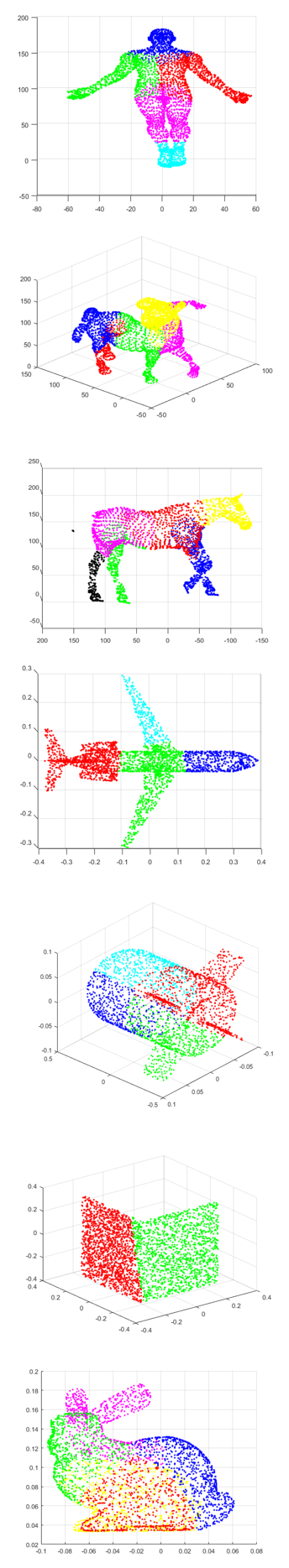}}
	\subfigure[GSP]{
		\label{mw21}
		\includegraphics[height=13cm]{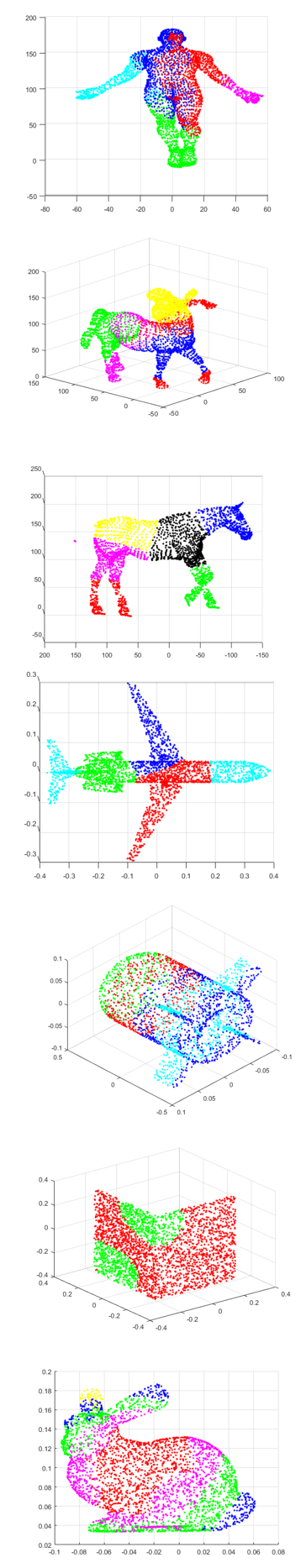}}
	\subfigure[Kmeans]{
		\label{mw41}
		\includegraphics[height=13cm]{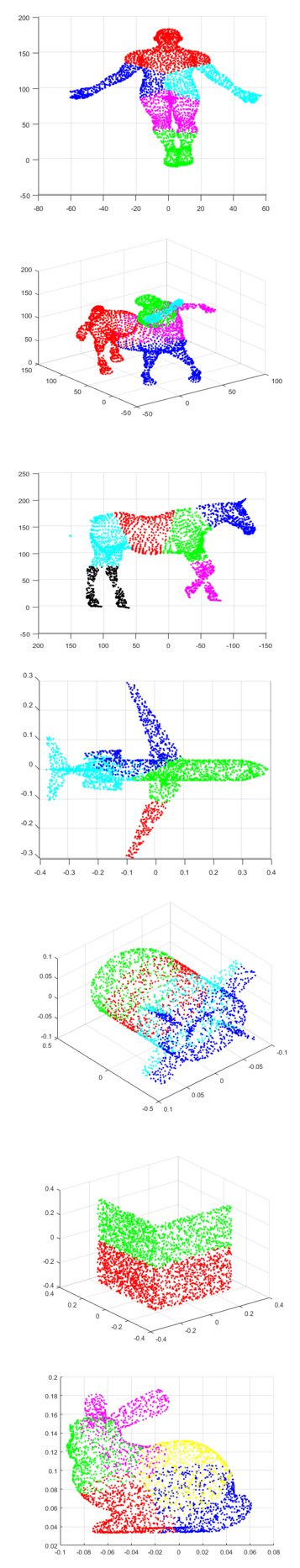}}
	\caption{Additional Results of Segmentation.}
	\label{mw111}
\end{figure}


\begin{thebibliography}{34}
\bibitem{c1} R. B. Rusu and S. Cousins, ``3d is here: Point cloud li-brary (pcl)," in \textit{2011 IEEE International Conference onRobotics and
Automation}, Shanghai, China, May 2011, pp. 1–4.

\bibitem{c2} A. Nguyen, and L. Bac, ``3D point cloud segmentation: a survey," in \textit{6th IEEE Conference on Robotics, Automation and
Mechatronics}, Manila, Philippines, Mar. 2013, pp. 225-230.

\bibitem{c3} J. Cheng, M. Qiao, W. Bian, and D. Tao, ``3D human posture segmentation by spectral clustering with surface normal
constraint", \textit{Signal Processing}, vol. 91, no. 9, pp. 2204-2212, Sep. 2011.

\bibitem{c4}  J. Biosca, and J. Lerma, ``Unsupervised robust planar segmentation of terrestrial laser scanner point clouds based on fuzzy
clustering methods", \textit{Journal of Photogrammetry and Remote Sensing}, vol.83, pp. 84-98, Sep. 2007.

\bibitem{c5} S. Filin, ``Surface clustering from airborne laser scanning data," \textit{International Archives of Photogrammetry, Remote Sensing
and Spatial Information Sciences}, pp.117124, 2001.

\bibitem{c6} T. Ma, Z. Wu, L. Feng, P. Luo, and X. Long, “Point cloud segmentation through spectral clustering,” in \textit{the 2nd International
Conference on Information Science and Engineering}, Hangzhou, China, Dec. 2010, pp. 1-4.

\bibitem{c7} H. Kisner, and U. Thomas, ``Segmentation of 3D point clouds using a new spectral clustering algorithm without a-priori knowledge," in \textit{VISIGRAPP}, vol.4, pp. 315-322, 2018.

\bibitem{c8} M. Geetha, and R. Rakendu, ``An improved method for segmentation of point cloud using minimum spanning tree," in \textit{2014 International Conference on Communication and Signal Processing}, Melmaruvathur, India, Apr. 2014, pp. 833-837.

\bibitem{c9} C. Wang, B. Samari, and K. Siddiqi, ``Local spectral graph convolution for point set feature learning," in \textit{Proceedings of the European Conference on Computer Vision (ECCV)}, Munich, Germany, Sep. 2018, pp. 52-66.
\bibitem{c10} S. Chen, D. Tian, C. Feng, A. Vetro and J. Kovacevic, ``Fast Resampling of Three-Dimensional Point Clouds via Graphs," in \textit{IEEE Transactions on Signal Processing}, vol. 66, no. 3, pp. 666-681, Feb, 2018.
\bibitem{gg1} G. Mateos, S. Segarra, A.G. Marques, and A. Ribeiro, ``Connecting the dots: Identifying network structure via graph signal processing," \textit{IEEE Signal Processing Magazine}, vol. 36, no. 3, pp. 16-43, May 2019.
\bibitem{gg} X. Dong, D. Thanou, M. Rabbat, and P. Frossard, ``Learning graphs from data: a signal representation perspective," \textit{IEEE Signal Processing Magazine}, vol. 36, no. 3, pp. 44-63, May 2019.
\bibitem{c11} A. Ortega, P. Frossard, J. Kovacevic, J. M. F. Moura, and P. Vandergheynst, ``Graph signal processing: overview,
challenges, and applications," in \textit{Proceedings of the IEEE}, vol. 106, no. 5, pp. 808-828, Apr. 2018.

\bibitem{c12} S. Zhang, Z. Ding, and S. Cui, “Introducing hypergraph signal processing: theoretical foundation and practical applications,” in \textit{IEEE Internet of Things Journal}, vol.7, no,1, pp. 639-660, Jan. 2020.

\bibitem{c13} S. Barbarossa, and S. Sardellitti, ``Topological signal processing over simplicial complexes," in \textit{IEEE Transactions on Signal Processing}, 2020.

\bibitem{c15} A. Afshar, J. C. Ho, B. Dilkina, I. Perros, E. B. Khalil, L. Xiong, and V.
Sunderam, ``Cp-ortho: an orthogonal tensor factorization framework for
spatio-temporal data," in \textit{Proceedings of the 25th ACM SIGSPATIAL International Conference on Advances in Geographic Information Systems},
Redondo Beach, CA, USA, Jan. 2017, p. 67.

\bibitem{c66} J. Pan, and M. K. Ng, ``Symmetric orthogonal approximation to symmetric tensors with applications to image reconstruction," \textit{Numerical Linear Algebra with Applications}, vol. 25, no. 5, e2180, Apr. 2018.


\bibitem{c99} T. G. Kolda, ``Orthogonal tensor decompositions," \textit{SIAM Journal on Matrix Analysis and Applications}, vol. 23, no. 1, pp. 243-255.

\bibitem{c14} A. Banerjee, C. Arnab, and M. Bibhash, ``Spectra of general hypergraphs," \textit{Linear Algebra and its Applications}, vol. 518, pp. 14-30, Dec. 2016.

\bibitem{c16} A. G. Marques, S. Segarra, G. Leus, and A. Ribeiro, ``Stationary graph processes and spectral estimation," in \textit{IEEE Transactions
on Signal Processing}, vol. 65, no. 22, pp. 5911-5926, Aug. 2017.

\bibitem{c17} B. Pasdeloup, V. Gripon, G. Mercier, D. Pastor, and M. G. Rabbat, `Characterization and inference of graph diffusion
processes from observations of stationary signals," \textit{IEEE Transactions on Signal and Information Processing over Networks},
vol. 4, no. 3, pp. 481-496, Aug. 2017.

\bibitem{c18} S. Zhang, S. Cui, and Z. Ding, ``Hypergraph spectral analysis and processing in 3D point cloud". arXiv preprint arXiv:2001.02384.

\bibitem{c19}  S. Segarra, A. G. Marques, G. Mateos, and A. Ribeiro, ``Network topology inference from spectral templates," \textit{IEEE
Transactions on Signal and Information Processing over Networks}, vol. 3, no. 3, pp. 467-483, Sep. 2017.

\bibitem{c20} X. Dong, D. Thanou, P. Frossard, and P. Vandergheynst, ``Learning Laplacian matrix in smooth graph signal representations,"
\textit{IEEE Transactions on Signal Processing}, vol. 64, no. 23, pp. 6160-6173, Dec. 2016


\bibitem{c31} Y. Gao, M. Wang, D. Tao, R. Ji, and Q. Dai, ``3-D object retrieval and recognition with hypergraph analysis," \textit{IEEE Transactions on Image Processing}, vol. 21, no. 9, pp. 4290-4303, Sep. 2012.

\bibitem{c32} D. Zhou, J. Huang, and B. Schölkopf, ``Learning with hypergraphs: clustering, classification, and embedding," in \textit{Advances in Neural Information Processing Systems}, Vancouver, Canada, Dec. 2007, pp. 1601-1608.
\bibitem{c33} L. Xiao, J. Wang, P. H. Kassani, Y. Zhang, Y. Bai, J. M. Stephen, T. W. Wilson, V. D. Calhoun, and Y. Wang, ``Multi-hypergraph learning-based brain functional connectivity analysis in fMRI data," \textit{IEEE Transactions on Medical Imaging}, vol. 39, no. 5, pp. 1746-1758, May 2020.

\bibitem{c34} J. Xu, J. E. Fowler, and L. Xiao, ``Hypergraph-regularized low-rank subspace clustering using superpixels for unsupervised spatial-spectral hyperspectral classification," \textit{IEEE Geoscience and Remote Sensing Letters}, Apr. 2020.


\bibitem{c21} A. M. Bronstein, M. M. Bronstein, and R. Kimmel, ``Numerical geometry of non-rigid shapes," Springer, 2008. ISBN: 978-0-387-73300-5.
\bibitem{c22} A. M. Bronstein, M. M. Bronstein, and R. Kimmel, ``Efficient computation of isometry-invariant distances between
surfaces”, \textit{SIAM J. Scientific Computing}, vol. 28, no. 5, pp. 1812-1836, 2006.
\bibitem{c23} A. M. Bronstein, M. M. Bronstein, U. Castellani, B. Falcidieno, A. Fusiello, A. Godil, L. J. Guibas, I. Kokkinos, Z. Lian,
M. Ovsjanikov, G. Patane, M. Spagnuolo, and R. Toldo, ``SHREC 2010: robust large-scale shape retrieval benchmark," 
\textit{Proc. EUROGRAPHICS Workshop on 3D Object Retrieval (3DOR)}, 2010.
\bibitem{c24} A. M. Bronstein, M. M. Bronstein, B. Bustos, U. Castellani, M. Crisani, B. Falcidieno, L. J. Guibas, I. Kokkinos, V.
Murino, M. Ovsjanikov, G. Patane, I. Sipiran, M. Spagnuolo, and J. Sun, ``SHREC 2010: robust feature detection and 
description benchmark," \textit{Proc. EUROGRAPHICS Workshop on 3D Object Retrieval (3DOR)}, 2010.

\bibitem{c25} L.Yi, and etl., ``Large-scale 3d shape reconstruction and segmentation from shapenet core55," arXiv preprint arXiv:1710.06104.

\bibitem{c26} A. X. Chang, and etl., ``Shapenet: An information-rich 3d model repository," arXiv preprint arXiv:1512.03012.

\bibitem{c27} U. V. Luxburg, ``A tutorial on spectral clustering," \textit{Statistics and Computing}, vol. 17, no. 4, pp. 395-416, Aug. 2007.

\end{thebibliography}
\end{document}